\documentclass[journal]{IEEEtran}

\usepackage{amsmath}
\usepackage{amsfonts}
\usepackage{amssymb}
\usepackage{amsthm}
\usepackage{tabularx}
\usepackage{mathrsfs} 

\usepackage{url}
\usepackage[colorlinks]{hyperref}
\newtheorem{theorem}{Theorem}
\newtheorem{remark}{Remark}
\newtheorem{definition}{Definition}

\newtheorem{corollary}{Corollary}

\usepackage{stfloats}
\usepackage{float}
\usepackage{graphicx}
\usepackage{cite}
\usepackage{xcolor}
\usepackage{subfigure}
\renewcommand{\vec}[1]{\mathbf{#1}}

\makeatletter
\def\blfootnote{\xdef\@thefnmark{}\@footnotetext}
\makeatother

\begin{document}
	
\title{\huge Copula-based Performance Analysis for Fluid Antenna Systems under Arbitrary Fading Channels\thanks{The work of Wong and Tong is supported by the Engineering and Physical Sciences Research Council (EPSRC) under Grant EP/W026813/1. For the purpose of open access, the authors will apply a Creative Commons Attribution (CC BY) licence to any Author Accepted Manuscript version arising. The work of L\'opez-Mart\'inez was funded in part by Consejer\'ia de Transformaci\'on
Econ\'omica, Industria, Conocimiento y Universidades of Junta de Andaluc\'ia, and in part by MCIN/AEI/10.13039/501100011033 through grants EMERGIA20 00297 and PID2020-118139RB-I00.}} 
\author{Farshad~Rostami~Ghadi\IEEEmembership{},~Kai-Kit~Wong, \IEEEmembership{Fellow}, \textit{IEEE},~\\
F. Javier~L\'opez-Mart\'inez, \IEEEmembership{Senior Member}, \textit{IEEE}, and Kin-Fai Tong, \IEEEmembership{Fellow}, \textit{IEEE}
%\vspace{-5mm}
}
\maketitle

\begin{abstract}
In this letter, we study the performance of a single-user fluid antenna system (FAS) under arbitrary fading distributions, in which the fading channel coefficients over the ports are correlated. We adopt copula theory to model the structure of dependency between fading coefficients. Specifically, we first derive an exact closed-from expression for the outage probability in the most general case, i.e., for any arbitrary choice of fading distribution and copula. Afterwards, for an important specific case, we analyze the performance of the outage probability under correlated Nakagami-$m$ fading channels by exploiting popular Archimedean copulas, namely, Frank, Clayton, and Gumbel. The results demonstrate that FAS outperforms the conventional single fixed-antenna system in terms of the outage probability. We also see that the spatial correlation dependency structure for the FAS is a key factor to determine its performance, which is natively captured through the choice of copula.
%the Frank copula representing the best spatial correlation structure for performance in FAS compared to other considered copulas. 
\end{abstract}

\begin{IEEEkeywords}
Fluid antenna system, arbitrary fading, correlation, outage probability, Archimedean copulas.
\end{IEEEkeywords}%\vspace{-3.5ex}

\maketitle

%\blfootnote{\noindent Copyright (c) 2015 IEEE. Personal use of this material is permitted. However, permission to use this material for any other purposes must be obtained from the IEEE by sending a request to pubs-permissions@ieee.org.} %Manuscript received January 25, 2021; revised XXX. The review of this paper was coordinated by XXXX. } 
		%This work has been funded in part by the Spanish Government and the European Fund for Regional Development FEDER (project  TEC2017-87913-R) and by Junta de Andalucia (project P18-RT-3175). 
\blfootnote{\noindent Farshad Rostami Ghadi and F.J. L\'opez Mart\'inez are with the Communications and Signal Processing Lab, Telecommunication Research Institute (TELMA), Universidad de M\'alaga, M\'alaga, 29010, (Spain). F.J. L\'opez Mart\'inez is also with the Department of Signal Theory, Networking and Communications, University of Granada, 18071, Granada (Spain) (e-mail: $\rm farshad@ic.uma.es$, $\rm fjlm@ugr.es$).}
\blfootnote{\noindent Kai-Kit Wong and Kin-Fai Tong are with the Department of Electronic and Electrical Engineering, University College London, London WC1E 6BT, United Kingdom. Kai-Kit Wong is also with Yonsei Frontier Lab, Yonsei University, Seoul, 03722, Korea.(e-mail: $\{\rm kai\text{-}kit.wong,k.tong\}@ucl.ac.uk$).}

\blfootnote{Digital Object Identifier 10.1109/XXX.2021.XXXXXXX}
%\IEEEpeerreviewmaketitle

\section{Introduction}\label{sec-intro}
Multiple-input multiple-output (MIMO) systems have been known as one of the most popular wireless technologies over the last decades, providing significant capacity by exploiting diversity over multiple signals undergoing independent fading. However, to ensure full diversity gain in MIMO systems, the antennas need to be separated by at least half the radiation wavelength, which is not always practical for mobile devices due to physical space limitations. To overcome such an issue, a novel fluid antenna system (FAS) has been recently proposed in \cite{wong2020fluid}, in which a single antenna has the ability to switch its location (i.e., ports) in a small space. This concept was greatly motivated by the recent advances in mechanically flexible antennas such as liquid metal antennas or ionized solutions as well as reconfigurable pixel-like antennas, e.g., \cite{dey2016microfluidically,singh2019multistate,song2013efficient}. 
    
%One of the main advantages of the FAS is that the location of the antenna is not necessarily fixed meaning that the antenna can freely be switched to a more desirable location if needed. However, the inherent spatial correlation in a small space between ports is a key challenge in the performance analysis of the FAS that need to be considered.

Several works have been recently conducted to investigate the performance of FAS from various viewpoints, e.g., \cite{wong2020performance,wong2023slow,tlebaldiyeva2022enhancing,skouroumounis2022large,wong2022closed,khammassi2023new}. The authors in \cite{wong2020performance} analyzed the ergodic capacity for a FAS under correlated Rayleigh fading channels, where they provided closed-form expressions of the capacity lower bound. An integral-form expression of the outage probability for a point-to-point FAS under correlated Nakagami-$m$ fading was derived in \cite{tlebaldiyeva2022enhancing}. Moreover, by exploiting stochastic geometry, \cite{skouroumounis2022large} derived an analytical expression of the outage probability in FAS large-scale cellular networks. Quite rightly, the performance of FAS is highly dependent on the spatial correlation model used for studying the performance of FAS. In \cite{wong2022closed}, it was revealed that previous contributions may not accurately capture the correlation between the FAS ports. For this purpose, \cite{khammassi2023new} proposed an eigenvalue-based model to approximate the spatial correlation given by Jake's model, where they illustrated that, under such model, the FAS has limited performance gain as the number of ports increases. Multiuser communications exploiting FAS has also been proposed recently in \cite{Wong-ffama2022,Wong-ffama2023,wong2023slow}.

%In \cite{wong2023slow}, the authors considered fluid antenna multiple access (FAMA) system and derived the signal-to-interference ratio (SIR) outage probability in a double integral-form expression. In addition, they derived a closed-form upper bound for the outage probability as well as an average outage capacity lower bound with an arbitrary number of interferers. 

Being able to correctly characterize the spatial correlation of the channel ports in FAS and yet maintaining mathematical tractability is challenging. Despite the previous efforts, there is lack of an accurate procedure to model the inherent channel correlation between ports. Specifically, generating the true multivariate distributions of correlated channels in FAS is hard due to mathematical and statistical limitations. To overcome this, one flexible approach to describe the structure of dependency between two or more random variables (RVs) is copula theory which has recently received significant attention in the performance analysis of wireless communication systems \cite{ghadi2020copula1,besser2020copula,besser2021fading,ghadi2022capacity,ghadi2022performance,trigui2022copula,ghadi2021role}. Generally speaking, copulas are mainly described with a dependence parameter which can measure the degree of dependence between RVs beyond linear correlation. Copulas can accurately generate the multivariate distributions of correlated RVs by only knowing the marginal distribution of each, and the copula parameter capturing the degree of dependence.

Motivated by the above, this letter analyzes the performance of FAS under arbitrary fading channels by exploiting the copula theory, which can accurately describe any sort of fading channel correlation. In contrast to previous works which only used approximations to generate specific multivariate distributions and then derived the outage probability in complicated integral forms, we propose general formulations for both multivariate distributions and outage probabilities in closed-form expressions. In particular, we first derive the cumulative distribution function (CDF) and probability density function (PDF) of the FAS channels in the most general case, i.e., for \textit{any} arbitrary choice of copula and fading distribution. Then, we obtain the closed-from expression of the outage probability in the FAS under correlated Nakgami-$m$ fading as an important special case by exploiting popular Archimedean\footnote{In this letter, we utilize Archimedean copulas for several reasons: (i) the ease with which they can be constructed; (ii) the great variety of families of copulas which belong to this class; and (iii) the many nice properties possessed by the members of this class \cite{nelsen2007introduction}.} copulas including Frank, Clayton, and Gumbel. Our analytical results show that there is no need to solve any complicated integrals to derive the multivariate distributions and the outage probability of the FAS. In addition, it is worth noting that the analytical results are valid for any arbitrary correlated fading distributions and can accurately quantify the degree of dependence between correlated fading channels. Furthermore, numerical results indicate that the best performance in terms of the outage probability occurs when the number of ports is large enough. The results also show that the Frank copula can be one of the best possible choices to analyze the channel correlation of the FAS.

\section{System Model}\label{sec-sys}
We consider a point-to-point communication system, where a single fixed-antenna transmitter sends an information signal $X$ with transmit power $P$ to a  receiver equipped with a single fluid antenna. We assume that the fluid antenna can move freely along $K$ equally distributed positions (i.e., ports) on a linear space. Assuming that there is only one RF chain in the FAS, only one port can be activated for communications, and the received signal at the $k$-th port can be expressed as
\begin{equation}
Y_k=h_kX+Z_k,
\end{equation}
where $h_k$ denotes the fading channel coefficient of the $k$-th port and $Z_k$ is the independent identically distributed (i.i.d.) additive white Gaussian noise (AWGN) with zero mean and variance $N$ at every port. In this scenario, we also assume that the FAS is able to select the best port with the strongest signal for communication, i.e., 
\begin{equation}
h_{\textrm{FAS}}=\max\left(|h_1|,|h_2|,\dots,|h_K|\right),
\end{equation}
in which the channel coefficients $h_k$ for $k\in\{1,2,...,K\}$ are correlated since they can be arbitrarily close to each other. Therefore, the received signal-to-noise ratio (SNR) for the FAS can be found as
\begin{equation}
\gamma=\frac{P|h_\mathrm{FAS}|^2}{N}=\bar{\gamma}|h_{\mathrm{FAS}}|^2,
\end{equation}
where $\bar{\gamma}=\frac{P}{N}$ is the average transmit SNR.

\section{Performance Analysis}
In order to obtain a closed-form expression of the outage probability, we first need to determine the distribution of $h_\mathrm{FAS}$ under arbitrary correlated fading coefficients. To proceed, we find it useful to briefly review some concepts of $d$-dimensional copula theory. Then, by exploiting the obtained distribution, we will be able to derive the closed-form expression of the outage probability in general and specific scenarios.

\subsection{Copula Properties}
\begin{definition}[$d$-dimensional copula]
Let $\vec{S}=(S_1,S_2,\dots,S_d)$ be a vector of $d$ RVs with marginal CDFs $F_{S_i}(s_i)$ for $i\in\{1,2,\dots,d\}$, respectively. Then, the corresponding joint CDF is defined as
\begin{multline}
F_{S_1,S_2,\dots,S_d}(s_1,s_2,\dots,s_d)\\
=\Pr(S_1\leq s_1,S_2\leq s_2,\dots,S_d\leq s_d).
\end{multline}
The copula function $C(u_1,u_2,\dots,u_d)$ of the random vector $\vec{S}$ defined on the unit hypercube $[0,1]^d$ with uniformly distributed RVs $U_i:=F_{S_i}(s_i)$ %for $i\in\{1,2,...,d\}$ 
over $[0,1]$ is given by
\begin{equation}
C(u_1,u_2,\dots,u_d)=\Pr(U_1\leq u_1,U_2\leq u_2,\dots,U_d\leq u_d),
\end{equation}
where $u_i=F_{S_i}(s_i)$.
\end{definition}

\begin{theorem}[Sklar's theorem]\label{thm-sklar}
Let $F_{S_1,S_2,...,S_d}(s_1,s_2,\dots,s_d)$ be a joint CDF of RVs with margins $F_{S_i}(s_i)$ for $i\in\{1,2,\dots,d\}$. Then, there exists one Copula function $C$ such that for all $s_i$ in the extended real line domain $\bar{R}$ \cite{nelsen2007introduction}
\begin{equation}\label{sklar}
F_{S_1,S_2,\dots,S_d}(s_1,\dots,s_d)=C\left(F_{S_1}(s_1),\dots,F_{S_d}(s_d)\right).
\end{equation}
\end{theorem}

%\begin{corollary}\label{col-pdf}
%	Let $f_{S_i}(s_i)$ be the margins of vector $\vec{S}$. Then, by applying the chain rule to Theorem \ref{thm-sklar}, the joint
%	probability density function (PDF) of vector $\vec{S}$ is given by
%\begin{align}\nonumber\label{pdf-d}
%	&f_{S_1,S_2,...,S_d}(s_1,s_2,...,s_d)=f_{S_1}(s_1)f_{S_2}(s_2)\;...\;f_{S_d}(s_d)\\
%	&\times c\big(F_{S_1}(s_1),F_{S_2}(s_2),...,F_{S_d}(s_d)\big),
%\end{align}
%where $c(.)$ denotes the copula density function and it can be determined as:
%\begin{align}\nonumber
%	&\hspace{-1cm}c\big(F_{S_1}(s_1),F_{S_2}(s_2),...,F_{S_d}(s_d)\big)\\
%	&=\frac{\partial^d C\left(F_{S_1}(s_1),F_{S_2}(s_2),...,F_{S_d}(s_d)\right)}{\partial{F_{S_1}(s_1)}\partial{F_{S_2}(s_2)}...\partial{F_{S_d}(s_d)}}.
%\end{align}
%\end{corollary}

\subsection{General Model: Arbitrary Correlated Fading Distribution}
Here, we derive the CDF and PDF of $h_\mathrm{FAS}$ as well as the outage probability for the most general case with any arbitrary choice of copula and fading distribution.

\begin{theorem}\label{thm-cdf-gen}
The CDF of $h_\mathrm{FAS}=\max\left(|h_1|,|h_2|,\dots,|h_K|\right)$ in the general dependence structure of arbitrary fading coefficients $|h_{k}|$ for $k\in\{1,2,\dots,K\}$ is given by
\begin{equation}\label{eq-cdf-gen}
F_{h_{\mathrm{FAS}}}(r)=C\left(F_{|h_1|}(r),F_{|h_2|}(r),\dots,F_{|h_K|}(r)\right),
\end{equation}
where $C(\cdot)$ is the copula function and $F_{|h_k|}(r)$ denotes the CDF of fading coefficient $|h_k|$ with an arbitrary distribution.
\end{theorem}

\begin{proof}
By exploiting the definition of the CDF, $F_{h_{\mathrm{FAS}}}(r)$ can be mathematically defined as
\begin{align}
F_{h_{\mathrm{FAS}}}(r)&\hspace{.5mm}=\Pr\left(h_\mathrm{FAS}\leq r\right)\notag\\
&\hspace{.5mm}=\Pr\left(\max\left\{|h_1|,|h_2|,\dots,|h_K|\right\}\leq r\right)\notag\\
&\hspace{.5mm}=\Pr\left(|h_1|\leq  r,|h_2|\leq r,\dots,|h_K|\leq r\right)\notag\\
&\hspace{.5mm}=F_{|h_1|,|h_2|,\dots,|h_K|}\left(r,r,\dots,r\right)\notag\\
&\overset{(a)}{=}C\left(F_{|h_1|}(r),F_{|h_2|}(r),\dots,F_{|h_K|}(r)\right),
\end{align}
where ($a$) is derived from Theorem \ref{thm-sklar}.
\end{proof}

\begin{theorem}\label{thm-pdf-gen}
The PDF of $h_\mathrm{FAS}=\max\left(|h_1|,|h_2|,\dots,|h_K|\right)$ in the general dependence structure of arbitrary fading coefficients $|h_{k}|$ for $k\in\{1,2,\dots,K\}$ is given by
\begin{multline}\label{eq-pdf-gen}
f_{h_{\mathrm{FAS}}}(r)=f_{|h_1|}(r)f_{|h_2|}(r)\cdots f_{|h_K|}(r)\\
\times c\left(F_{|h_1|}(r),F_{|h_2|}(r),\dots,F_{|h_K|}(r)\right),
\end{multline}
where $f_{|h_k|}(r)$ denotes the marginal PDF of fading coefficient $|h_k|$ with an arbitrary distribution and  $c(\cdot)$ is the copula density function which can be determined as
\begin{multline}
c\left(F_{|h_1|}(r),F_{|h_2|}(r),\dots,F_{|h_K|}(r)\right)\\
=\frac{\partial^d C\left(F_{|h_1|}(r),F_{|h_2|}(r),\dots,F_{|h_K|}(r)\right)}{\partial{F_{|h_1|}(r)}\partial{F_{|h_2|}(r)}\cdots \partial{F_{|h_K|}(r)}}.
\end{multline}
\end{theorem}

\begin{proof}
The result is obtained by applying the chain rule to Theorem \ref{thm-cdf-gen}.
%The proof is directly obtained form Theorem \ref{col-pdf}.
\end{proof}

\begin{theorem}\label{thm-op-gen}
The outage probability of the considered FAS in the general dependence structure under an arbitrary fading distribution is given by
\begin{equation}\label{eq-out-gen}
P_\mathrm{out}=C\left(F_{|h_1|}(\hat{\gamma}),F_{|h_2|}(\hat{\gamma}),\dots,F_{|h_K|}(\hat{\gamma})\right),
\end{equation}
where $\hat{\gamma}=\sqrt\frac{\gamma_{\mathrm{th}}}{\bar{\gamma}}$.
\end{theorem}

\begin{proof}
Outage probability is an appropriate metric to evaluate the performance of FAS, which is defined as the probability that the random SNR $\gamma$ is less than an SNR threshold $\gamma_\mathrm{th}$. Therefore, the outage probability is defined as
\begin{align}
P_{\mathrm{out}}&=\Pr\left(\gamma\leq\gamma_{\mathrm{th}}\right)=\Pr\left(h_\mathrm{FAS}\leq\sqrt\frac{\gamma_{\mathrm{th}}}{\bar{\gamma}}\right)\notag\\
&=\Pr\left(\max\left\{|h_1|,|h_2|,\dots,|h_K|\right\}\leq \sqrt\frac{\gamma_{\mathrm{th}}}{\bar{\gamma}}\right)\notag\\
&=F_{h_{\mathrm{FAS}}}(\hat{\gamma}),
\end{align}
where by exploiting the obtained CDF form \eqref{eq-cdf-gen}, the proof is completed. 
\end{proof}

\begin{remark}
In contrast to previous contributions, the results in \eqref{eq-cdf-gen}, \eqref{eq-pdf-gen}, and \eqref{eq-out-gen} indicate that there is no need to solve any complicated integral for deriving the outage probability and the joint distribution of channel coefficients in the FAS, provided that the copula function $C(\cdot)$ is given in closed-form. 
\end{remark}

Remarkably, the results in Theorems \ref{thm-cdf-gen}--\ref{thm-op-gen} are valid for \textit{any} arbitrary choice of fading distribution and copula function over the proposed FAS. Now, to analyze the system performance, we consider a special case in the following section.

\subsection{Special Case: Correlated Nakagami-$m$ Fading}
Here, for exemplary purposes, we assume that the fading channel coefficients $|h_k|$ follow the Nakagami-$m$ distribution, where the parameter $m\ge0.5$ denotes the fading severity. In the special case $m=1$, the Rayleigh fading with an exponentially distributed instantaneous power is recovered. Hence, the marginal distributions for the fading channel coefficient $|h_k|$ can be expressed as
\begin{equation}\label{eq-pdf-nak}
f_{|h_k|}(r)=\frac{2m^m}{\Gamma(m)\mu^m}r^{2m-1}\mathrm{e}^{-\frac{m}{\mu}r^2},
\end{equation}
with the following CDF
\begin{equation}\label{eq-cdf-nak}
	F_{|h_k|}(r)=\frac{\gamma\left(m,\frac{m}{\mu}r^2\right)}{\Gamma(m)},
\end{equation}
in which $m$ and $\mu$ define the shape and spread parameters, respectively. The terms $\Gamma(\cdot)$ and $\gamma(\cdot,\cdot)$ denote gamma function and the lower incomplete gamma function, respectively.

In order to evaluate the structure of dependency beyond linear coordination between the correlated fading channel coefficients, there are many types of copula functions that can be used. However, in this letter, we derive the analytical expressions by using the popular Archimedean copulas, namely, Frank, Clayton, and Gumbel. These flexible copula functions can describe both weak and strong correlations between correlated RVs under positive dependence structures. 
%between correlated random variables. In addition, it should be noted that the FGM copula can only be used for weak correlation, whereas the Frank copula can be utilized for both strong and weak correlation.
%\begin{definition}[FGM copula]
%The generalized FGM copula $C_{\mathrm{FGM}}$ of $d$-dimension is defined as:
%	\begin{align}\nonumber
%		&C_{\mathrm{FGM}}(u_1,u_2,...,u_d)\\
%		&=u_1u_2...u_d\left(1+\sum_{l=2}^{d}\;\sum_{1\leq j_1<...<j_l\leq d}\theta_{j_1j_2...j_d}\bar{u}_{j_1}\bar{u}_{j_2}...\bar{u}_{j_d}\right),
%	\end{align}
%where $\bar{u}=1-u$ and $\theta\in[-1,1]$ is a dependence structure parameter of FGM Copula. For $\theta\in[-1,0)$ and $\theta\in(0,1]$ the structure of dependency is negative and positive, respectively. Besides, $\theta=0$ denotes the independent case. 
%	\end{definition}

\begin{definition}[Frank copula]
The generalized Frank copula $C_{\mathrm{FR}}$ of $d$-dimension is defined as
\begin{equation}\label{eq-frank}
C_{\mathrm{FR}}(u_1,u_2,\dots,u_d)=-\frac{1}{\alpha}\ln\left(1+\frac{\prod_{j=1}^d\left(\mathrm{e}^{-\alpha u_j}-1\right)}{\mathrm{e}^{-\alpha}-1}\right),
\end{equation}
where $\alpha\in \mathbb{R}\backslash\{0\}$ is a dependence structure parameter of Frank copula. The independent case is achieved when $\alpha\rightarrow 0$.
\end{definition}

\begin{definition}[Clayton copula]
The generalized Clayton copula $C_{\mathrm{CL}}$ of $d$-dimension is defined as
\begin{equation}\label{eq-cl}
C_{\mathrm{CL}}(u_1,u_2,\dots,u_d)=\left[\sum_{j=1}^d \left(u_j^{-\beta}-1\right)+1\right]^{-\frac{1}{\beta}},
\end{equation}
where $\beta\in [0,\infty)$ is a dependence structure parameter of Clayton copula. The independent case is achieved if $\beta=0$.
\end{definition}

\begin{definition}[Gumbel copula]
The generalized Gumbel copula $C_{\mathrm{GU}}$ of $d$-dimension is defined as
\begin{equation}\label{eq-gu}
C_{\mathrm{GU}}(u_1,u_2,\dots,u_d)=\exp\left(-\left[\sum_{j=1}^d \left(-\ln u_j\right)^\theta\right]^{\frac{1}{\theta}}\right),
\end{equation}
where $\theta\in [1,\infty)$ is a dependence structure parameter of Gumbel copula. The independent case is achieved if $\theta=1$.
\end{definition}

Now, by exploiting the definition of the above-mentioned copulas, the outage probability in the specific model can be determined in the following theorem. 

\begin{theorem}\label{thm-op}
The outage probability of the considered FAS under Nakagami-$m$ fading channel, using the Frank, Clayton, and Gumbel copulas is, respectively, given by
\begin{align}\label{eq-out-fr}
P_{\mathrm{out}}^{\mathrm{FR}}=-\frac{1}{\alpha}\ln\left(1+\frac{\left[\exp\left(\frac{-\alpha\gamma\left(m,\frac{m}{\mu}{\hat{\gamma}}^2\right)}{\Gamma(m)}\right)-1\right]^K}{\mathrm{e}^{-\alpha}-1}\right),
\end{align}
\begin{align}\label{eq-out-cl}
P_{\mathrm{out}}^{\mathrm{CL}}=\left[K \left(\left(\frac{\gamma\left(m,\frac{m}{\mu}{\hat{\gamma}}^2\right)}{\Gamma(m)}\right)^{-\beta}-1\right)+1\right]^{-\frac{1}{\beta}},
\end{align}
and
\begin{align}\label{eq-out-gu}
P_{\mathrm{out}}^{\mathrm{GU}}=\exp\left(K^{\frac{1}{\theta}} \ln \frac{\gamma\left(m,\frac{m}{\mu}{\hat{\gamma}}^2\right)}{\Gamma(m)}\right).
\end{align}
\end{theorem}

\begin{proof}
By inserting \eqref{eq-cdf-nak} into \eqref{eq-frank}, \eqref{eq-cl}, and \eqref{eq-gu} for $u_j=F_{|h_k|}(\hat{\gamma})$ and then considering $\eqref{eq-out-gen}$, the proof is completed.  
\end{proof}

\begin{corollary}
As $K\rightarrow\infty$, the outage probability goes to $0$ as long as $\alpha, \beta\neq 0$.
\end{corollary}

\begin{proof}
It can be seen from \eqref{eq-out-fr} that the fraction term inside the logarithm is the product of $K$ less-than-one values. If $K\rightarrow \infty$, it goes to $0$. In \eqref{eq-out-cl}, it is straightforward as $K\rightarrow \infty$, the outage probability reaches $0$ due to the term $\frac{-1}{\beta}$. Finally, in \eqref{eq-out-gu}, the logarithm value is always negative, thereby when $K\rightarrow\infty$, the outage probability achieves $0$. 
\end{proof}

\begin{remark}
The outage probability for the considered FAS can be accurately obtained in a closed-form expression according to Theorem \ref{thm-op}. We can see that the outage probability highly depends on the number of ports $K$, dependence parameters $\alpha$, $\beta$, and $\theta$, and the fading parameter $m$, namely, the performance of the outage probability will improve as $K$ and $m$ increase and  $\alpha$, $\beta$, and $\theta$ decrease. 
\end{remark}

%	\begin{figure}[!t]
%	\centering
%	\includegraphics[width=0.9\columnwidth]{system.jpg}
%	\caption{System model of a downlink STAR-RIS-aided NOMA communication.}
%	\label{system}
%\end{figure}
	%\subsection{SNR Distribution}
\begin{figure}[!t]
	\centering
	\includegraphics[width=0.9\columnwidth]{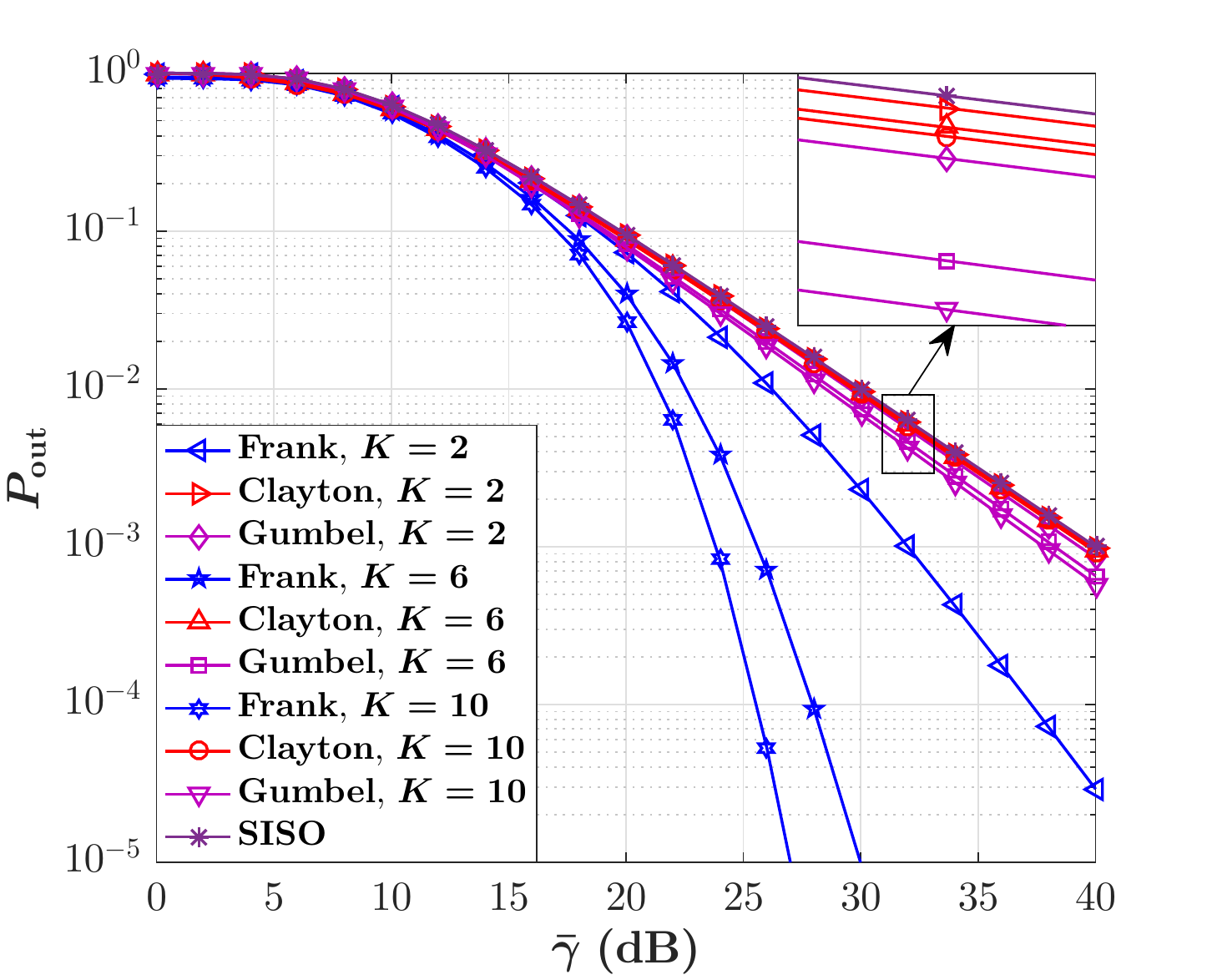}
	\caption{Outage probability versus average SNR for selected values of $K$, when $m=1$, $\alpha=\beta=\theta=30$, and $\mu=1$.}%\vspace{-0.5cm}
	\label{fig-out-k}
\end{figure}
\begin{figure}[!t]
	\centering
	\includegraphics[width=0.9\columnwidth]{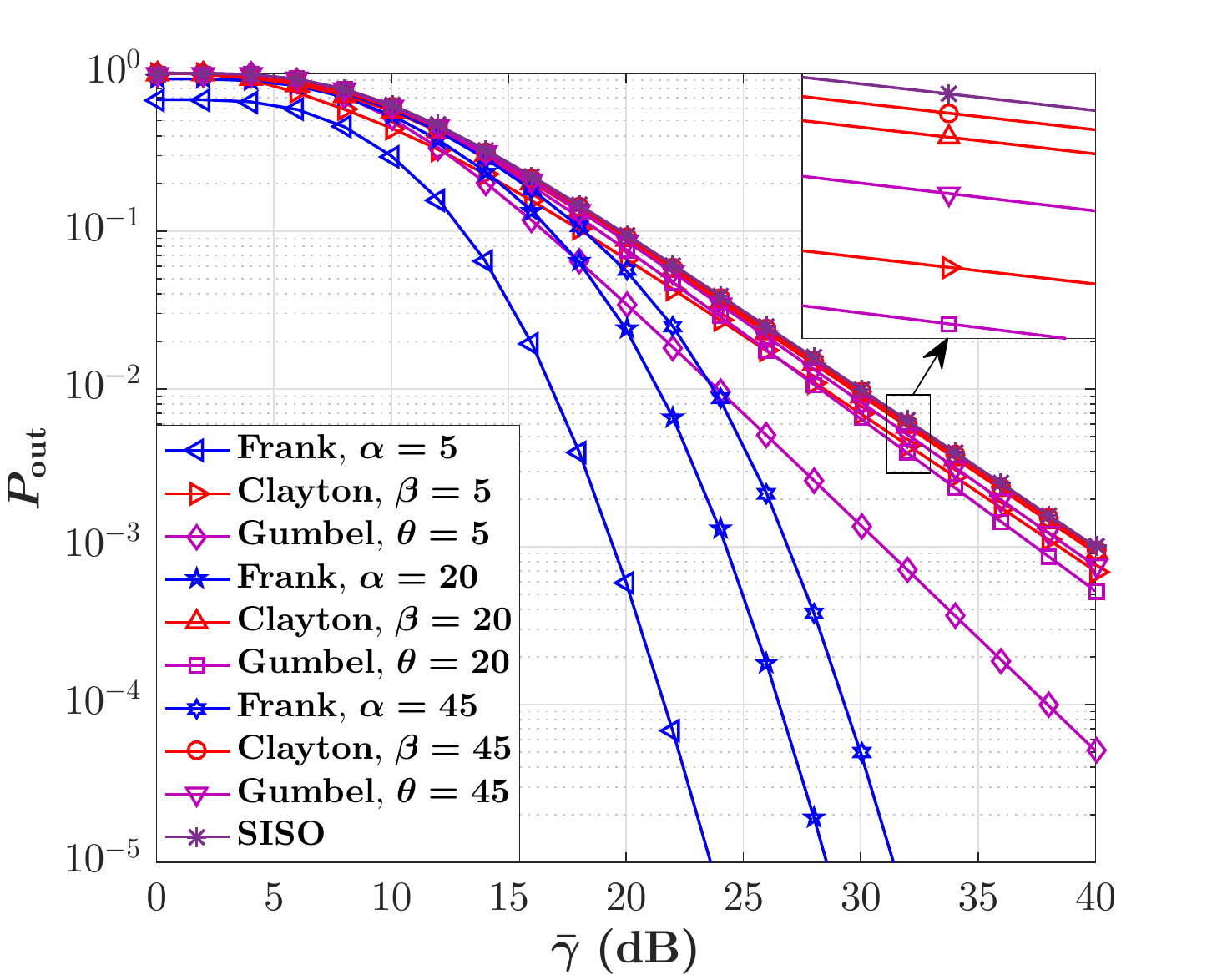}
	\caption{Outage probability versus average SNR $\bar{\gamma}$ for selected values of dependence parameters, when $m=1$, $K=6$, and $\mu=1$.}%\vspace{-0.45cm}
	\label{fig-dep}
\end{figure}
\begin{figure}[!t]
	\centering
	\includegraphics[width=0.9\columnwidth]{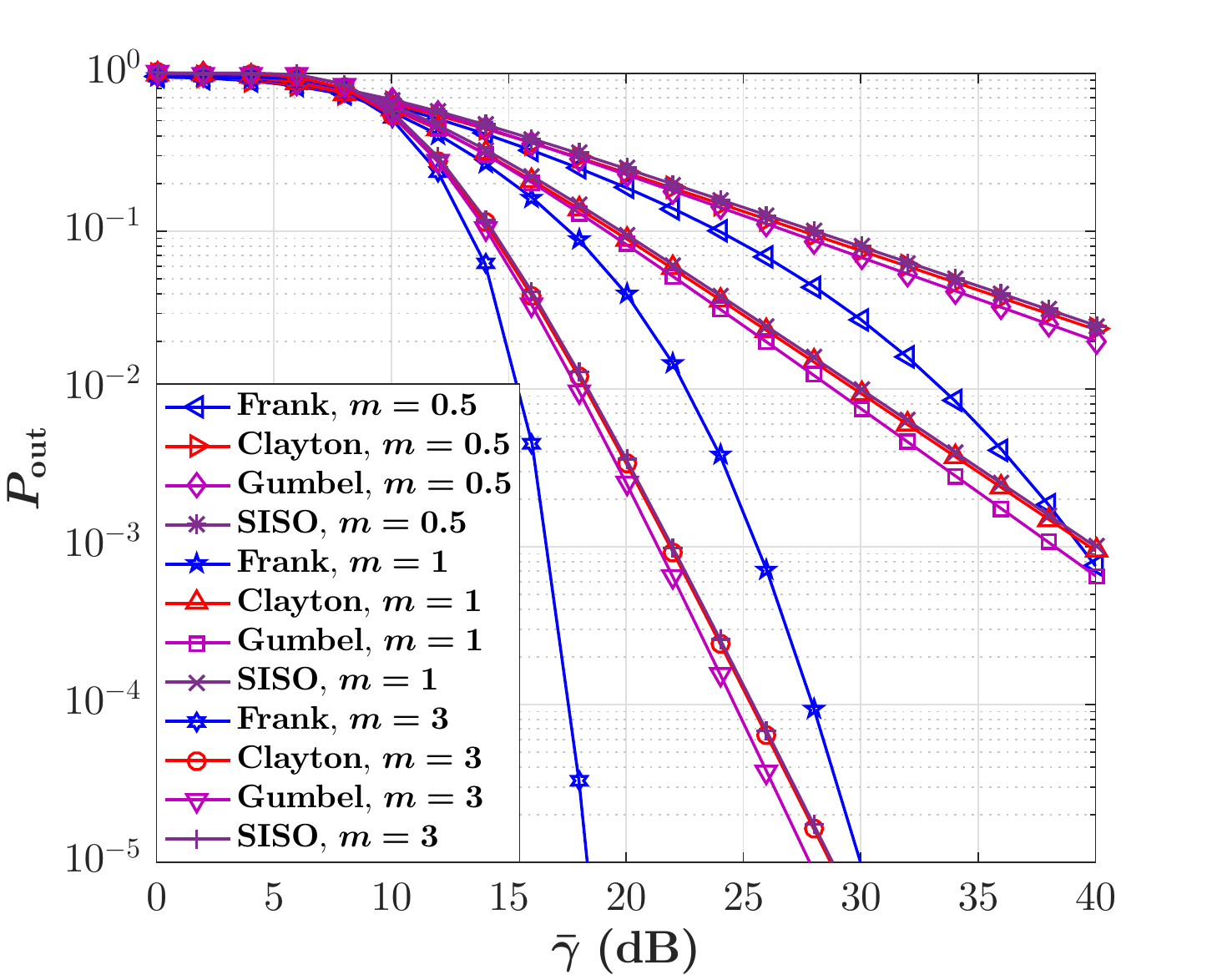}
	\caption{Outage probability versus average SNR $\bar{\gamma}$ for selected values of the fading parameter $m$, when $K=6$, $\alpha=\beta=\theta=30$, and $\mu=1$.}%\vspace{-0.7cm}
	\label{fig-m}
\end{figure}

\section{Numerical Results}\label{sec-num}
In this section, we present numerical results to evaluate the FAS performance in terms of the outage probability. Fig.~\ref{fig-out-k} indicates the behavior of the outage probability against the average SNR $\bar{\gamma}$ for different numbers of fluid antenna ports $K$ under correlated fading channels. It can be observed that FAS provides lower outage probability compared with the conventional single-input-single-output (SISO) system so that for a fixed value of dependence parameters, the outage probability performance improves as $K$ increases. By careful observation of the curves, we can also see that the Frank copula is more sensitive to even a small change of $K$, compared with the Clayton and Gumbel copulas. Frank copula also appears to give the least outage probability in FAS compared to others, which suggests that one should design FAS that resembles the correlation structure described by Frank copula.

The impact of fading channel correlation on the performance of the outage probability is illustrated in Fig.~\ref{fig-dep}. It is clearly seen that fading correlation has destructive effects on the outage probability. As the dependence parameters of the Archimedean copulas grow, the outage probability increases. However, we can see that even under a strong positive dependence structure between fading channel coefficients, FAS outperforms the SISO system in terms of the outage probability. Moreover, it is found that the Frank copula also provides the best performance from a correlation perspective compared with two other considered copulas. 

Regarding the importance of fading severity on the performance of FAS, Fig.~\ref{fig-m} shows the behavior of the outage probability versus the average SNR for different values of the fading parameter $m$ under correlated Nakagami-$m$ distribution. The results indicate that as the fading severity reduces (i.e., $m$ increases), the outage probability performance ameliorates and such an improvement is more noticeable when Frank copula is considered to describe the fading correlation. To gain more insight into the impact of the number of ports on the FAS performance, Fig.~\ref{fig-k} shows the behavior of the outage probability in terms of $K$ for selected values of the average SNR. It can be seen that even under a strong correlation provided by the Frank copula (e.g., $\alpha=30$), the outage probability significantly decreases as $K$ becomes large. Fig.~\ref{fig-dep2} also shows the impact of the dependence parameter on the outage probability performance for the considered FAS. The results reveal that when the correlation is weak (i.e., low dependence parameter), a lower outage probability is achieved. In addition, the Frank copula provides better performance for weak correlation compared with the Clayton and Gumbel copulas. However, as the correlation between channel coefficients becomes stronger (i.e., high dependence parameter), the Clayton and Gumbel copulas offer almost the same behavior as the Frank copula. Therefore, by comparing the results in Figs.~\ref{fig-k} and \ref{fig-dep2}, it can be found that the best performance for the outage probability occurs when $K$ is sufficiently large and the dependence parameter is small enough.

\begin{figure}[!t]
	\centering
	\includegraphics[width=0.9\columnwidth]{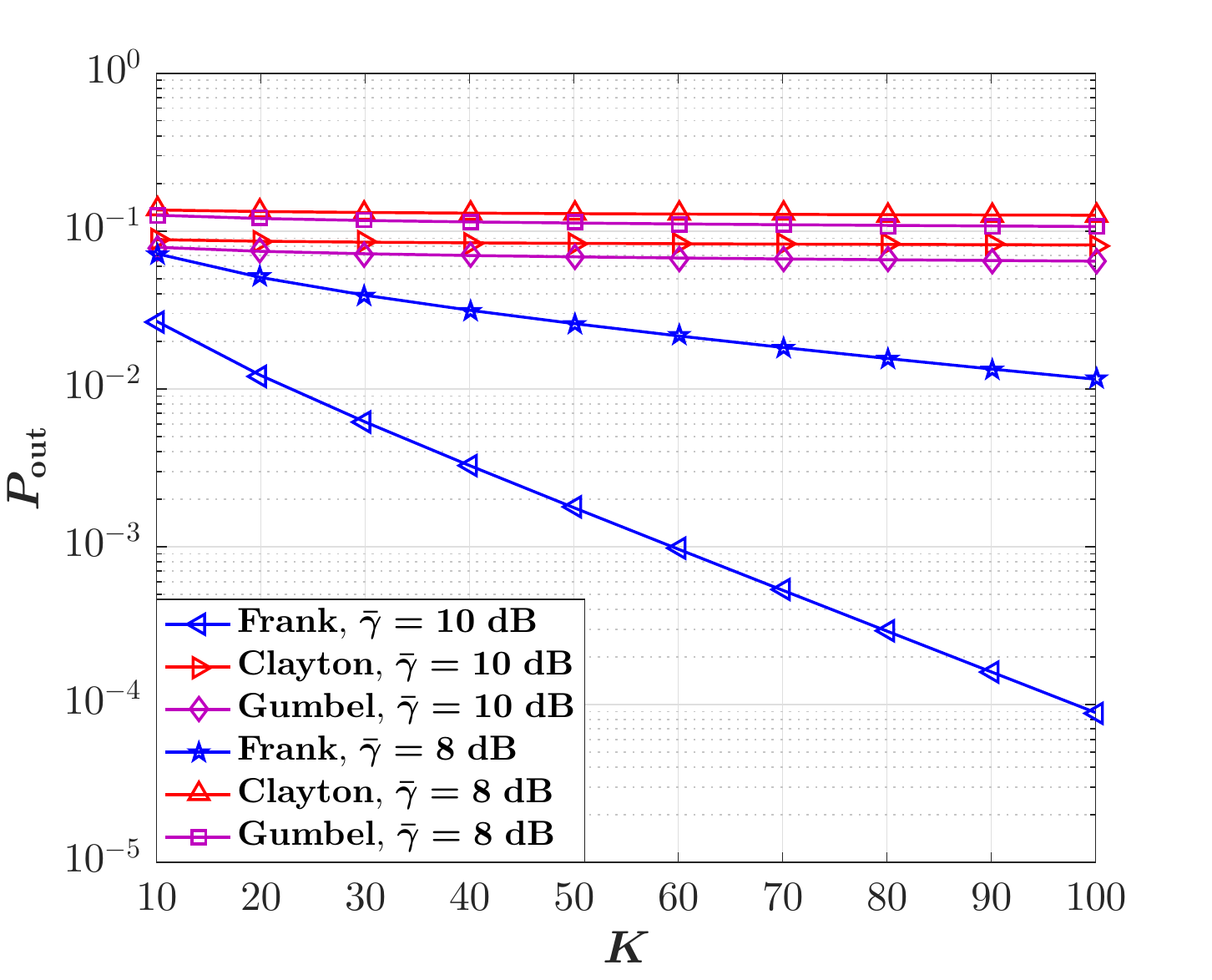}
	\caption{Outage probability versus number of ports $K$ for selected values of the average SNR, when $m=1$, $\alpha=\beta=\theta=30$, and $\mu=1$.}%\vspace{-0.45cm}
	\label{fig-k}
\end{figure}
\begin{figure}[!t]
	\centering
	\includegraphics[width=0.9\columnwidth]{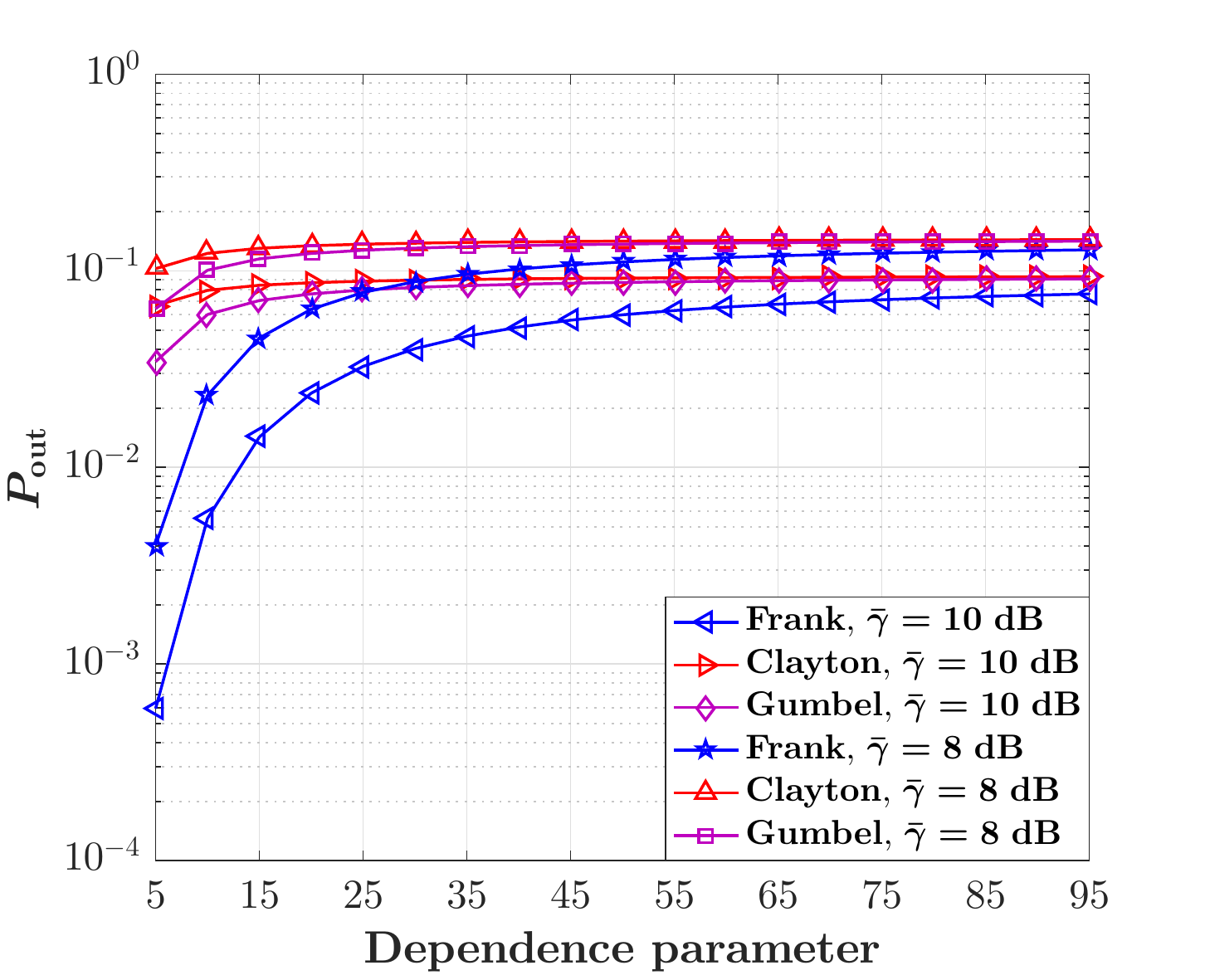}
	\caption{Outage probability versus dependence parameters for selected values of the average SNR, when $m=1$, $K=6$, and $\mu=1$.}%\vspace{-0.5cm}
	\label{fig-dep2}
\end{figure}

\section{Conclusion}\label{sec-con}
In this letter, we studied the performance of a point-to-point FAS, where the correlated fading channels have arbitrary distributions. First, by exploiting the copula approach to model the structure of dependency between correlated arbitrary fading channels, we derived the closed-form expression for the outage probability in the most general case. Then, as a specific scenario, we analyzed the outage probability for the considered FAS under correlated Nakagami-$m$ fading channels by using popular Archimedean copulas. The results showed that the best performance in terms of the outage probability for the considered FAS is achieved by increasing the number of fluid antenna ports and the fading parameter as well as reducing the value of the dependence parameter. Additionally, it was shown that the Frank copula provides a deeper insight in order to evaluate the effect of the FAS parameters compared with other proposed copulas.

%\vspace{-0.5cm}
\bibliographystyle{IEEEtran}
%\bibliography{sample.bib}


\end{document}